\documentclass[a4paper,reqno]{amsart}

\usepackage[T1]{fontenc}
\usepackage[utf8]{inputenc}
\usepackage[english]{babel}
\usepackage{amsmath,amssymb,mathtools}
\usepackage{enumitem}
\usepackage{xcolor}
\usepackage[linktocpage=true,colorlinks=true,linkcolor=blue,citecolor=red,urlcolor=teal]{hyperref}

\newcommand{\C}{\mathbb C}

\newcommand{\PP}{\mathbb P}

\newcommand{\wt}{\operatorname{wt}}
\newcommand{\Conf}{\operatorname{Conf}}
\newcommand{\Rem}{\operatorname{Rem}}
\newcommand{\inw}{\operatorname{in}_{\!w}}

\theoremstyle{plain}
\newtheorem{theorem}{Theorem}[section]
\newtheorem{proposition}[theorem]{Proposition}
\newtheorem{lemma}[theorem]{Lemma}
\newtheorem{corollary}[theorem]{Corollary}

\theoremstyle{definition}

\theoremstyle{remark}
\newtheorem{remark}[theorem]{Remark}

\numberwithin{equation}{section}

\title[A generalized Stieltjes system]{A generalized Stieltjes system with polynomial source}

\author[D. Masoero]{D. Masoero}
\address{Grupo de F\'isica Matem\'atica, Universidade de Lisboa, Edif\'icio C6, Campo Grande, Lisboa, Portugal}
\email{dmasoero@gmail.com}

\author[B. Shapiro]{B. Shapiro}
\address{Department of Mathematics, Stockholm University, SE-106 91 Stockholm, Sweden}
\email{shapiro@math.su.se}


\begin{document}

\begin{abstract}
Let $Q$ be a monic polynomial of degree $M+1$. We study the algebraic system
\[
        \sum_{j\ne i}\frac{1}{x_i-x_j}=Q(x_i),\qquad i=1,\ldots,N,
\]
for pairwise distinct complex numbers $x_1,\ldots,x_N$, modulo permutations of
these numbers.  The case $M=0$ is, after a translation, the classical Stieltjes
system for the zeros of a Hermite polynomial.  We prove that, for arbitrary
$Q$, the number of solutions is at most $\binom{N+M}{N}$, and that the
coefficient equations for the associated monic Stieltjes polynomial have total
intersection multiplicity exactly $\binom{N+M}{N}$.  Consequently the bound is
attained for all $Q$ in a non-empty Zariski open subset of the affine space of
monic polynomials of degree $M+1$.  We also describe the solutions when the
coefficient of the linear term of $Q$ is large: the system splits into $M+1$
weakly coupled classical Stieltjes systems, one near each zero of $Q$.
\end{abstract}

\maketitle

\section{Introduction}

The classical Heine--Stieltjes theory relates equilibrium configurations of
movable logarithmic charges to polynomial solutions of second order linear
differential equations; see Heine \cite{Heine}, Stieltjes \cite{Stieltjes},
and the modern survey \cite{ShapiroHS}.  The Hermite equation is a basic
degenerate example of the same general philosophy, and its electrostatic
interpretation is discussed, for instance, in Szeg\H{o}'s book
\cite[Sec.~6.7]{Szego} and in \cite{DimitrovShapiro}.  The present note treats
a simple polynomial-source version of this degenerate situation.

Fix integers $N\geq1$ and $M\geq0$.  For a monic polynomial
\begin{equation}\label{eq:source}
        Q(z)=z^{M+1}+a_Mz^M+\cdots+a_1z+a_0
\end{equation}
of degree $M+1$, consider
\begin{equation}\label{eq:genStieltjes}
        \sum_{j\ne i}\frac{1}{x_i-x_j}=Q(x_i),\qquad i=1,\ldots,N,
\end{equation}
with $(x_1,\ldots,x_N)\in\Conf_N(\C)$, where
\[
        \Conf_N(\C):=\{(x_1,\ldots,x_N):x_i\ne x_j\text{ for }i\ne j\}.
\]
Two ordered solutions are considered equivalent if they differ by a permutation
of the $x_i$'s.  Equivalently, an equivalence class is represented by the monic
polynomial
\[
        P(z)=\prod_{i=1}^N(z-x_i).
\]
We shall therefore freely pass between unordered solutions of
\eqref{eq:genStieltjes} and monic polynomials $P$ of degree $N$.

\begin{theorem}\label{thm:main}
For every monic polynomial $Q$ of degree $M+1$, the number of inequivalent
solutions of \eqref{eq:genStieltjes} is at most
\[
        \binom{N+M}{N}.
\]
More precisely, the associated coefficient equations for the monic polynomial
$P$ have total intersection multiplicity exactly $\binom{N+M}{N}$.  Hence the
same binomial coefficient is the exact number of inequivalent solutions for all
$Q$ in a non-empty Zariski open subset of the affine space of monic polynomials
of degree $M+1$.
\end{theorem}

The proof has two parts.  First, we rewrite \eqref{eq:genStieltjes} as a
Heine--Stieltjes divisibility condition and compute the length of the resulting
coefficient scheme by a weighted B\'ezout argument.  Second, we exhibit an
explicit non-empty open set by analysing the limit in which the coefficient
$a_1$ of the linear term tends to infinity.  In this limit the zeros of $Q$
consist of $M$ large zeros and one small zero, and each solution of
\eqref{eq:genStieltjes} is obtained by distributing the $N$ movable zeros among
these $M+1$ centres.  After the natural rescaling near each centre, one obtains
the standard Hermite/Stieltjes system.

\section{From roots to a polynomial equation}

Let $P(z)=\prod_{i=1}^{N}(z-x_i)$ have simple roots.  The elementary identity
\begin{equation}\label{eq:logder-id}
        \frac{P''(x_i)}{2P'(x_i)}=\sum_{j\ne i}\frac{1}{x_i-x_j}
\end{equation}
shows that \eqref{eq:genStieltjes} is equivalent to
\[
        P''(x_i)-2Q(x_i)P'(x_i)=0,
        \qquad i=1,\ldots,N.
\]
Thus $P$ divides $P''-2QP'$.  This gives the following standard form of the
Heine--Stieltjes correspondence.

\begin{lemma}\label{lem:divisibility}
Let $Q$ be as in \eqref{eq:source}.  Equivalence classes of solutions of
\eqref{eq:genStieltjes} are in bijection with monic polynomials $P$ of degree
$N$ for which there exists a polynomial $V$ of degree at most $M$ such that
\begin{equation}\label{eq:HS}
        P''(z)-2Q(z)P'(z)-V(z)P(z)=0.
\end{equation}
For such a polynomial $P$, all roots of $P$ are simple.
\end{lemma}

\begin{proof}
If $(x_1,\ldots,x_N)$ is a solution of \eqref{eq:genStieltjes}, then
\eqref{eq:logder-id} shows that $P''-2QP'$ vanishes at every root of $P$.
Since the roots are simple, $P$ divides $P''-2QP'$.  The quotient has degree at
most $M$, because $Q$ has degree $M+1$ and $P$ has degree $N$.

Conversely, suppose that \eqref{eq:HS} holds.  It remains first to observe that
$P$ cannot have a multiple root.  If
$P(z)=(z-\xi)^m\widetilde P(z)$ with $m\ge2$ and $\widetilde P(\xi)\ne0$, then
$P''$ has order $m-2$ at $\xi$ with leading coefficient
$m(m-1)\widetilde P(\xi)$, while $QP'$ has order at least $m-1$ and $VP$ has
order at least $m$.  The coefficient of $(z-\xi)^{m-2}$ in the left-hand side
of \eqref{eq:HS} is therefore non-zero, a contradiction.

Thus the roots of $P$ are simple.  Substituting any root $\xi$ of $P$ in
\eqref{eq:HS} gives $P''(\xi)-2Q(\xi)P'(\xi)=0$, and
\eqref{eq:logder-id} gives \eqref{eq:genStieltjes}.
\end{proof}

Let
\[
        P(z)=z^N+c_1z^{N-1}+\cdots+c_N.
\]
Divide $P''-2QP'$ by $P$ and denote by
\begin{equation}\label{eq:remainder}
        \Rem_Q(P)=r_1(c)z^{N-1}+r_2(c)z^{N-2}+\cdots+r_N(c)
\end{equation}
the remainder, where $c=(c_1,\ldots,c_N)$.  By Lemma~\ref{lem:divisibility},
solutions of \eqref{eq:genStieltjes} are exactly the common zeros of
\begin{equation}\label{eq:coefficient-system}
        r_1(c)=\cdots=r_N(c)=0.
\end{equation}

Assign the weights
\[
        \wt(z)=1,
        \qquad \wt(c_j)=j,\quad j=1,\ldots,N.
\]
Then $P(z)=z^N+c_1z^{N-1}+\cdots+c_N$ is weighted homogeneous of total weight
$N$.  For a polynomial $f(c)$, write $\inw(f)$ for its highest weighted
homogeneous part.

\begin{lemma}\label{lem:weighted-leading}
For $j=1,\ldots,N$, the polynomial $r_j$ has weighted degree at most $M+j$.
Its highest weighted homogeneous part is the coefficient of $z^{N-j}$ in the
remainder of $-2z^{M+1}P'(z)$ modulo $P(z)$.  These highest parts have no
common zero in $\C^N\setminus\{0\}$.
\end{lemma}

\begin{proof}
If a weighted homogeneous polynomial $G(z,c)$ has total weight $N+d$, then the
remainder of $G$ modulo the weighted homogeneous monic polynomial $P$ has the
form
\[
        h_1(c)z^{N-1}+\cdots+h_N(c),
        \qquad \wt(h_j)=d+j.
\]
Indeed, Euclidean division by the monic polynomial $P$ is compatible with the
weight filtration.  The term $-2z^{M+1}P'$ has total weight $N+M$.  The terms
coming from $P''$ and from the lower degree part of $Q$ have strictly smaller
weight.  This proves the first two assertions.

Assume now that all highest parts vanish at some $c\in\C^N$.  Then the
remainder of $z^{M+1}P'$ modulo $P$ is zero, i.e.
\begin{equation}\label{eq:leading-divisibility}
        P\mid z^{M+1}P'.
\end{equation}
If $P$ has a non-zero root $\alpha$ of multiplicity $m\ge1$, then the
right-hand side of \eqref{eq:leading-divisibility} has multiplicity $m-1$ at
$\alpha$, which is impossible.  Therefore every root of $P$ is zero, so
$P=z^N$ and $c_1=\cdots=c_N=0$.  Hence the highest parts have no common zero
away from the origin.
\end{proof}

We shall use the following weighted form of B\'ezout's theorem.  The proof is
included to make clear that the multiplicity counted below is the affine length
of the coefficient scheme.

\begin{lemma}[weighted B\'ezout]\label{lem:weighted-bezout}
Let $f_1,\ldots,f_N\in\C[c_1,\ldots,c_N]$, with $\wt(c_j)=w_j>0$, have weighted
degrees at most $d_1,\ldots,d_N$.  Suppose that their highest weighted
homogeneous parts $g_i:=\inw(f_i)$ have no common zero in
$\C^N\setminus\{0\}$.  Then the common zero scheme of $f_1,\ldots,f_N$ in
$\C^N$ is finite and has length
\[
        \frac{d_1\cdots d_N}{w_1\cdots w_N}.
\]
\end{lemma}

\begin{proof}
The polynomials $g_1,\ldots,g_N$ form a weighted homogeneous system of
parameters in the polynomial ring $S=\C[c_1,\ldots,c_N]$.  Since $S$ is
Cohen--Macaulay, they form a regular sequence.  Hence
\[
        \operatorname{Hilb}_{S/(g_1,\ldots,g_N)}(t)
        =\frac{\prod_{i=1}^N(1-t^{d_i})}{\prod_{j=1}^N(1-t^{w_j})},
\]
and the length of $S/(g_1,\ldots,g_N)$ is the value at $t=1$ after cancelling
the zero of numerator and denominator, namely
$\prod_i d_i/\prod_j w_j$.

It remains to pass from the highest parts $g_i$ to $f_i$.  We use the
following elementary filtered-graded fact.  If $A$ is a filtered algebra, $f\in A$
has initial form $g\in\operatorname{gr}A$, and $g$ is a non-zero-divisor on
$\operatorname{gr}A$, then
$\operatorname{gr}(A/(f))\simeq (\operatorname{gr}A)/(g)$; this follows by
comparing highest weighted terms.  Applying this fact successively to the
regular sequence $g_1,\ldots,g_N$ shows that $f_1,\ldots,f_N$ is a regular
sequence, that $S/(f_1,\ldots,f_N)$ is finite-dimensional, and that its
associated graded algebra is $S/(g_1,\ldots,g_N)$.  Thus the two quotients have
the same length.

Equivalently, one may homogenize the equations with a new variable $\tau$ of
weight $1$ and apply B\'ezout in the weighted projective space
$\PP(w_1,\ldots,w_N,1)$; see, for example, \cite{Dolgachev}.  The hypothesis
precisely says that the homogenized system has no solutions on the hyperplane
$\tau=0$.
\end{proof}

\begin{proposition}\label{prop:length}
For every monic polynomial $Q$ of degree $M+1$, the coefficient system
\eqref{eq:coefficient-system} has finitely many zeros and total intersection
multiplicity
\[
        \binom{N+M}{N}.
\]
\end{proposition}

\begin{proof}
Apply Lemma~\ref{lem:weighted-bezout} to
$r_1,\ldots,r_N$.  By Lemma~\ref{lem:weighted-leading}, the weights are
$w_j=j$ and the degrees are $d_j=M+j$.  Therefore the total multiplicity is
\[
        \prod_{j=1}^{N}\frac{M+j}{j}
        =\binom{N+M}{N}.
\]
\end{proof}

\section{The classical Stieltjes system}

Let $\mathsf H_n$ denote the monic Hermite polynomial, normalized by
\begin{equation}\label{eq:monic-Hermite}
        \mathsf H_n(z)=2^{-n}(-1)^n e^{z^2}\frac{d^n}{dz^n}e^{-z^2}.
\end{equation}
It satisfies
\begin{equation}\label{eq:Hermite-ODE}
        \mathsf H_n''(z)-2z\mathsf H_n'(z)+2n\mathsf H_n(z)=0.
\end{equation}
Its zeros are real, simple, and symmetric with respect to the origin.  We write
them as
\[
        u_1^{(n)}<\cdots<u_n^{(n)},
        \qquad u_i^{(n)}=-u_{n+1-i}^{(n)}.
\]
For $n=0$ the set of zeros is empty.

\begin{lemma}\label{lem:classical}
The standard Stieltjes system
\begin{equation}\label{eq:standard-Stieltjes}
        \sum_{j\ne i}\frac{1}{t_i-t_j}=t_i,
        \qquad i=1,\ldots,n,
\end{equation}
has, modulo permutations, the unique solution given by the zeros of
$\mathsf H_n$.
\end{lemma}

\begin{proof}
If $R(z)=\prod_{i=1}^n(z-t_i)$ corresponds to a solution, then $R''-2zR'$ is
divisible by $R$.  Since $R$ is monic of degree $n$, comparison of leading
terms gives
\[
        R''-2zR'+2nR=0.
\]
The monic polynomial solution of \eqref{eq:Hermite-ODE} is unique by the usual
coefficient recursion; hence $R=\mathsf H_n$.
\end{proof}

\begin{lemma}\label{lem:Jacobian}
Let $J^{(n)}$ be the Jacobian matrix of the map
\[
        F_i(t)=t_i-\sum_{j\ne i}\frac{1}{t_i-t_j},
        \qquad i=1,\ldots,n,
\]
at the Hermite zero configuration $u^{(n)}=(u_1^{(n)},\ldots,u_n^{(n)})$.
Then the eigenvalues of $J^{(n)}$ are $1,2,\ldots,n$.  In particular,
$J^{(n)}$ is invertible.
\end{lemma}

\begin{proof}
For $m=0,\ldots,n-1$, consider
\[
        R_\varepsilon(z)=\mathsf H_n(z)+\varepsilon\mathsf H_m(z)
\]
and let $t_i(\varepsilon)$ be the root of $R_\varepsilon$ near
$u_i^{(n)}$.  Differentiating
$R_\varepsilon(t_i(\varepsilon))=0$ at $\varepsilon=0$ gives
\[
        t_i'(0)=-\frac{\mathsf H_m(u_i^{(n)})}{\mathsf H_n'(u_i^{(n)})}.
\]
Using \eqref{eq:Hermite-ODE} for $\mathsf H_n$ and $\mathsf H_m$, one obtains,
at the roots of $R_\varepsilon$,
\[
        F_i(t(\varepsilon))
        =-(n-m)\varepsilon\,
          \frac{\mathsf H_m(t_i(\varepsilon))}{R_\varepsilon'(t_i(\varepsilon))}.
\]
Differentiating at $\varepsilon=0$ gives
\[
        J^{(n)}t'(0)=(n-m)t'(0).
\]
Thus the eigenvalues are $n,n-1,\ldots,1$, and the corresponding eigenvectors
are linearly independent because the eigenvalues are distinct.
\end{proof}

For later use, note also the translated form of the case $M=0$.  If
$Q(z)=z+a_0$, then the unique solution of \eqref{eq:genStieltjes} is
\[
        x_i=u_i^{(N)}-a_0,
        \qquad i=1,\ldots,N,
\]
up to permutation, and it is reduced by Lemma~\ref{lem:Jacobian}.

\section{Large linear coefficient}

In this section assume $M\ge1$.  Fix all coefficients of $Q$ except the linear
one, and write them as $a_0,a_2,\ldots,a_M$; if $M=1$, this means only $a_0$ is
fixed.  Put
\begin{equation}\label{eq:Qnu}
        Q_\nu(z)=z^{M+1}+\sum_{\ell=2}^{M}a_\ell z^\ell-\nu^{2M}z+a_0.
\end{equation}
Thus $a_1=-\nu^{2M}$, and the limit $|a_1|\to\infty$ is studied through
$\nu\to\infty$.

\begin{lemma}\label{lem:source-roots}
For $|\nu|$ sufficiently large, the polynomial $Q_\nu$ has $M$ simple large
zeros
\[
        \alpha_k(\nu)=\omega^k\nu^2+O(1),
        \qquad k=0,\ldots,M-1,
        \qquad \omega=e^{2\pi i/M},
\]
and one simple small zero
\[
        \gamma(\nu)=a_0\nu^{-2M}+O(\nu^{-4M}).
\]
More precisely, $\nu^{-2}\alpha_k(\nu)$ and $\nu^{2M}\gamma(\nu)$ are analytic
functions of $\nu^{-1}$ near $\nu^{-1}=0$, after the choice of the branch of
$\nu$.  Moreover
\begin{equation}\label{eq:derivatives-centres}
        Q_\nu'(\alpha_k(\nu))=M\nu^{2M}\bigl(1+O(\nu^{-2})\bigr),
        \qquad
        Q_\nu'(\gamma(\nu))=-\nu^{2M}\bigl(1+O(\nu^{-2})\bigr).
\end{equation}
\end{lemma}

\begin{proof}
For the large zeros set $z=\nu^2y$.  Dividing $Q_\nu(\nu^2y)$ by
$\nu^{2M+2}$ gives
\[
        y^{M+1}-y+O(\nu^{-2})
\]
uniformly for $y$ in compact sets.  The non-zero roots of $y^{M+1}-y$ are
$1,\omega,\ldots,\omega^{M-1}$, and the derivative of $y^{M+1}-y$ at each of
these roots is $M$.  The analytic implicit function theorem gives the $M$
large roots.  For the small zero set $z=\nu^{-2M}y$; then
\[
        Q_\nu(\nu^{-2M}y)=a_0-y+O(\nu^{-4M}),
\]
which gives the stated small root.  The derivative estimates follow by
substitution in $Q_\nu'$.
\end{proof}

Choose analytic square roots and define
\begin{equation}\label{eq:local-scales}
        s_k(\nu):=Q_\nu'(\alpha_k(\nu))^{-1/2},
        \qquad
        s_O(\nu):=Q_\nu'(\gamma(\nu))^{-1/2},
\end{equation}
with asymptotics
\begin{equation}\label{eq:scale-asymptotics}
        s_k(\nu)=M^{-1/2}\nu^{-M}(1+O(\nu^{-2})),
        \qquad
        s_O(\nu)=i\nu^{-M}(1+O(\nu^{-2})).
\end{equation}
The sign choices are immaterial, since the Hermite zero configuration is
invariant under $t\mapsto -t$ up to permutation.

A weak composition of $N$ into $M+1$ parts will be written
\[
        \lambda=(\lambda_0,\ldots,\lambda_{M-1},\lambda_O),
        \qquad |\lambda|=\sum_{k=0}^{M-1}\lambda_k+\lambda_O=N.
\]
There are $\binom{N+M}{N}$ such compositions.  Given such a $\lambda$, we place
$\lambda_k$ variables near $\alpha_k$ and $\lambda_O$ variables near $\gamma$:
\begin{align}\label{eq:cluster-coordinates}
        x_{k,i}&=\alpha_k(\nu)+s_k(\nu)t_{k,i},
        && i=1,\ldots,\lambda_k,
        \quad k=0,\ldots,M-1,\\
        x_{O,i}&=\gamma(\nu)+s_O(\nu)t_{O,i},
        && i=1,\ldots,\lambda_O.\nonumber
\end{align}
Empty blocks are omitted.

\begin{proposition}\label{prop:large-linear}
Fix $a_0,a_2,\ldots,a_M$ and a weak composition $\lambda$ of $N$ into $M+1$
parts.  For $|\nu|$ sufficiently large, the Stieltjes system
\eqref{eq:genStieltjes} with source $Q_\nu$ has a unique solution, up to
permutations inside the blocks of \eqref{eq:cluster-coordinates}, of the form
\begin{equation}\label{eq:cluster-asymptotics}
        t_{k,i}=u_i^{(\lambda_k)}+O(\nu^{-1}),
        \qquad
        t_{O,i}=u_i^{(\lambda_O)}+O(\nu^{-1}),
\end{equation}
where $u_i^{(r)}$ are the zeros of the monic Hermite polynomial $\mathsf H_r$.
The corresponding solution is reduced as a zero of the coefficient system
\eqref{eq:coefficient-system}.
\end{proposition}

\begin{proof}
We prove the assertion for ordered variables inside each non-empty block; the
quotient by the block permutation group is then immediate.

Consider a block centred at a zero
$\rho_b\in\{\alpha_0,\ldots,\alpha_{M-1},\gamma\}$, with scale
$s_b=Q_\nu'(\rho_b)^{-1/2}$.  In that block write
$x_{b,i}=\rho_b+s_bt_{b,i}$.  Since $Q_\nu(\rho_b)=0$ and
$s_b^2Q_\nu'(\rho_b)=1$, Taylor's formula gives, uniformly for $t$ in compact
sets,
\begin{equation}\label{eq:source-local-expansion}
        s_b Q_\nu(\rho_b+s_bt)=t+O(\nu^{-1}).
\end{equation}
In fact the error is smaller than $O(\nu^{-1})$, but this bound is sufficient.
For two variables in the same block,
\[
        \frac{s_b}{x_{b,i}-x_{b,j}}
        =\frac{1}{t_{b,i}-t_{b,j}}.
\]
For variables in distinct blocks the centres differ by order $\nu^2$, while
$s_b=O(\nu^{-M})$; hence every cross-interaction, after multiplication by
$s_b$, is $O(\nu^{-M-2})=O(\nu^{-1})$.

Multiplying the Stieltjes equation for $x_{b,i}$ by $s_b$ therefore gives a
system analytic in the local variables and in $\nu^{-1}$ near $\nu^{-1}=0$.
Its limiting system is the product, over all blocks, of
\begin{equation}\label{eq:block-limit}
        t_{b,i}-\sum_{j\ne i}\frac{1}{t_{b,i}-t_{b,j}}=0.
\end{equation}
By Lemma~\ref{lem:classical}, the limiting solution in the $b$-th block is the
Hermite zero set of size equal to that block.  By Lemma~\ref{lem:Jacobian}, the
Jacobian of the limiting product system is block diagonal with invertible
blocks.  The analytic implicit function theorem therefore gives a unique
analytic branch with the asymptotics \eqref{eq:cluster-asymptotics}.

The same Jacobian argument shows that the obtained zero is reduced.  Passing
from ordered block variables to the coefficient system does not introduce an
extra multiplicity, because the roots in the limiting Hermite configurations
are distinct and the passage from unordered roots to coefficients is locally
biholomorphic away from the big diagonal.
\end{proof}

\begin{corollary}\label{cor:large-a1-exact}
For every fixed $a_0,a_2,\ldots,a_M$ and all sufficiently large $|a_1|$, with
$a_1=-\nu^{2M}$, the system \eqref{eq:genStieltjes} has exactly
$\binom{N+M}{N}$ inequivalent solutions.  All of them are reduced.
\end{corollary}

\begin{proof}
Proposition~\ref{prop:large-linear} constructs one reduced solution for each
weak composition of $N$ into $M+1$ parts.  These solutions are distinct because
they have different cluster multiplicities around the zeros of $Q_\nu$.  There
are $\binom{N+M}{N}$ such compositions.  Proposition~\ref{prop:length} says
that the total multiplicity of the whole coefficient scheme is exactly the same
number, so no further solutions are possible.
\end{proof}

\begin{proof}[Proof of Theorem~\ref{thm:main}]
The upper bound and the total multiplicity statement are precisely
Proposition~\ref{prop:length}, together with Lemma~\ref{lem:divisibility}.
For $M=0$ the translated Hermite solution described above is reduced.  For
$M\ge1$, Corollary~\ref{cor:large-a1-exact} gives at least one monic source
$Q$ for which all zeros of the coefficient system are reduced and their number
is $\binom{N+M}{N}$.  Since reducedness of the fibres of a finite algebraic
family is a Zariski open condition, equality holds on a non-empty Zariski open
subset of the affine coefficient space of monic sources $Q$.
\end{proof}

\begin{remark}
The proof of Proposition~\ref{prop:large-linear} explains why the binomial
coefficient appears geometrically: in the large-$|a_1|$ regime, the $N$ roots
of the Stieltjes polynomial choose one of the $M$ large zeros of $Q$ or the one
small zero of $Q$, and inside each chosen cluster the only possible limiting
configuration is the Hermite configuration of the corresponding size.
\end{remark}

\section*{Acknowledgements}
D. Masoero is partially supported by FCT Grant ``The Nonlinear Stokes
Phenomenon. A unifying perspective on Integrable Models, Enumerative Geometry,
and Special Functions'', 2021.00091.CEECIND.

\end{document}